\documentclass[submission,copyright,creativecommons]{eptcs}
\usepackage{breakurl}             
\usepackage{underscore}

\usepackage{latexsym}
\usepackage{color}
\usepackage{lscape}
\usepackage{setspace}
\usepackage{mymacros}

\usepackage{makeidx}
\usepackage{epsfig}
\usepackage{accenti}
\usepackage{path}
\usepackage{indentfirst}
\usepackage{graphicx, multicol, latexsym, amsmath, amssymb}
\usepackage{blindtext}

\usepackage{amsthm}
\usepackage{amsmath}
\usepackage{amsfonts}
\usepackage{amssymb}
\usepackage{stmaryrd}
\usepackage{booktabs}
\usepackage{graphicx}
\usepackage{subfig}
\usepackage{caption}
\usepackage{algorithm}
\usepackage{algpseudocode}
\usepackage{tikz-cd}
\usetikzlibrary{matrix,arrows,decorations.pathmorphing}

\newcounter{ncomm}

\newcommand{\com}[1]{{\color{black}#1}}

\begin{document}

\newtheorem{example}{Example}
\newtheorem{definition}{Definition}
\newtheorem{theorem}{Theorem}
\newtheorem{lemma}{Lemma}
\newtheorem{proposition}{Proposition}
\newtheorem{remark}{Remark}
\newtheorem{corollary}{Corollary}
\newtheorem{claim}{Claim}

\title{Persistent Stochastic  Non-Interference}
\author{Jane Hillston
\institute{University of Edinburgh, UK}
\email{Jane.Hillston@ed.ac.uk}
\and
Carla Piazza
\institute{Universit\`a di Udine, Italy}
\email{carla.piazza@uniud.it}
\and
 Sabina Rossi
 \institute{Universit\`a Ca' Foscari Venezia, Italy}
 \email{sabina.rossi@unive.it}
}
\def\titlerunning{Persistent Stochastic  Non-Interference}
\def\authorrunning{J. Hillston, C. Piazza \& S. Rossi}

\maketitle

\begin{abstract}
In this paper we present an information flow security property for  stochastic, cooperating, processes expressed as terms of the Performance Evaluation Process Algebra (PEPA). We introduce the notion of \emph{Persistent Stochastic Non-Interference (PSNI)} based on the idea that every state reachable by a process satisfies a basic \emph{Stochastic Non-Interference (SNI)} property. 
The structural operational semantics of PEPA allows us to give two characterizations of \emph{PSNI}: the first involves a single bisimulation-like equivalence check, while the second is formulated in terms of unwinding conditions. The observation equivalence at the base of our definition relies on the notion of lumpability and ensures that, for a secure process $P$,  the steady state probability of observing the system being in a specific state $P'$ is independent from its possible high level interactions.
\end{abstract}

\section{Introduction}
\label{sec:intro}

Non-Interference is an information flow security property
which aims at protecting sensitive data from undesired accesses. In particular, it consists in protecting the confidentiality of information by guaranteeing that high level, sensitive, information never flows to low level, unauthorized, users. It is well known that access control policies or cryptographic protocols are, in general, not sufficient to forbid unwanted flows which may arise from the so called covert channels or from some weakness in the cryptographic algorithms.

The notion of Non-Interference for deterministic systems has been introduced in  \cite{NI}  and it has been extended to  non-deterministic systems. Non-Interference has been then studied in different settings such as programming languages \cite{FRS05,SM03,SV98},
trace models \cite{Man00b,McL94},
cryptographic protocols \cite{ABF18,BR05,FGM-ICALP},
process calculi \cite{CR06,CR07,FG01,HR03,RS01},
probabilistic models \cite{AB09,PHW02}, timed models \cite{JSAC,GLM03}, and stochastic models \cite{AB09}.

In this paper we study a notion of Non-Interference  for  stochastic, cooperating, processes expressed as terms of the Performance Evaluation Process Algebra (PEPA) \cite{hillston:book}.
We introduce the notion of \emph{Persistent Stochastic Non-Interference (PSNI)} based on the idea that every state reachable by a process satisfies a basic \emph{Stochastic Non-Interference (SNI)} property.
By imposing that security persists during process execution, the system  is guaranteed 
to be dynamically secure in the sense that every potential transition leads the process to a secure state. 
Property \emph{SNI} is inspired by the \emph{Bisimulation-based Non-Deducibility on Compositions (BNDC)} property defined in \cite{FG95} for non-deterministic CCS processes.
In our setting, the definition has the following form:
a process $P$ is secure if a low level observer cannot distinguish the behavior of $P$ in isolation from  the behavior of $P$ cooperating with any possible high level process $H$. The notion of observation that we consider is based on the concept of lumpability  for the underlying Markov chain \cite{marin:valuetools13,marin:mascots14,MR-acta17}.
 Formally, property \emph{SNI} is defined as: for any high level process $H$ which may enable only high level activities, 
$$P\setminus {\cal H}\approx_l (P\sync{{\cal H}}H)/{\cal H}$$
where $P\setminus {\cal H}$ represents the low level view of $P$ in isolation, while $(P\sync{{\cal H}}H)/{\cal H}$ denotes the low level view of $P$ interacting with the high process $H$. The observation equivalence $\approx_l$ is the \emph{lumpable bisimilarity}  defined in \cite{marin:valuetools13} which is a
characterization
of a lumpable
relation over the terms of the process algebra PEPA preserving contextuality and inducing a lumping in the underlying Markov processes.
Notice that this basic security property, that we call \emph{Stochastic Non-Interference (SNI)} is not persistent in the sense that it is not preserved during system execution.
Thus, it might happen that a system satisfying \emph{SNI} reaches a state which is not secure.
To overcome this problem we introduce the notion of \emph{Persistent Stochastic Non-Interference (PSNI)} which requires that every state reachable by the system is secure, i.e.,
$P$ is secure if and only if
$$\forall P' \mbox{ reachable from } P,  \ \  P' \mbox{ satisfies } SNI\,.$$
Notice that this property contains two universal quantifications: one over all the reachable states and another one, inside the definition of \emph{SNI}, over all the possible high level processes which may interact with the considered system. The main contributions of this paper are:
\begin{itemize}
\itemsep -0.5pt
\item we provide a characterization of \emph{PSNI} in terms of a single bisimulation-based check thus avoiding the universal quantification over all the high level contexts;
\item based on the structural operational semantics of PEPA, we provide a characterization of \emph{PSNI} expressed in terms of unwinding conditions;
  \item we prove that \emph{PSNI} is compositional with respect to low prefix, cooperation over low actions and hiding;
  \item we prove that if $P$ is secure then the equivalence class $[P]$ with respect to  lumpable bisimilarity $\approx_l$ is closed under \emph{PSNI};
  \item we show through an example that if $P$ is secure then, from the low level point of view,
the steady state probability of observing the system being in a specific state $P'$ is independent from the possible high level interactions of $P$.
    
\end{itemize}

\emph{Structure of the paper.} 
The paper is organized as follows: in Section \ref{sec:calculus} we introduce the process algebra PEPA, its structural operational semantics, and the observation equivalence named \emph{lumpable bisimilarity}. The notion of \emph{Persistent Stochastic Non-Interference (PSNI)} and its characterizations are presented in Section \ref{sec:psni}. In Section \ref{sec:compositionality} we prove some compositionality result and other properties of \emph{PSNI}. Comparisons with other SOS-based persistent security properties are discussed in Section \ref{sec:comparison}. Finally, Section \ref{sec:conclusion} concludes the paper.
\begin{table*}[t]
	\begin{center}
	\begin{tabular}{ccc}
		\toprule
		& \\
 \multicolumn{3}{c}
{$\dfrac{}{(\alpha,r).P \transits{(\alpha,r)} P}$ \qquad
$\dfrac{P \transits{(\alpha,r)} P'}{P+Q \transits{(\alpha,r)} P'}$ \qquad
  $\dfrac{Q \transits{(\alpha,r)} Q'}{P+Q \transits{(\alpha,r)} Q'}$}
  \\[8mm]
  \multicolumn{3}{c}
 { $\dfrac{P \transits{(\alpha,r)} P'}{P/L \transits{(\alpha,r)} P'/L}$  $(\alpha  \not \in L)$ \qquad
$\dfrac{P \transits{(\alpha,r)} P'}{P/L \transits{(\tau,r)} P'/L}$  $(\alpha  \in L)$ }
\\[8mm]
\multicolumn{3}{c}
{$\dfrac{P \transits{(\alpha,r)} P'}{A \transits{(\alpha,r)} P'}$  $(A\rmdef P)$ \qquad $\dfrac{P \transits{(\alpha,r)} P'}{P \sync{L} Q \transits{(\alpha,r)} P'\sync{L} Q}$  $(\alpha \not \in L)$ \qquad
 $\dfrac{Q \transits{(\alpha,r)} Q'}{P \sync{L} Q \transits{(\alpha,r)} P\sync{L} Q'}$ $(\alpha \not \in L)$}\\[8mm]
 \multicolumn{3}{c}
{ $\dfrac{P \transits{(\alpha,r_1)} P' \ \ Q \transits{(\alpha,r_2)} Q'}{P\sync{L} Q \transits{(\alpha,R)} P'\sync{L} Q'}$ \ \ \   $R =  \dfrac{r_1}{r_{\alpha}(P)} \dfrac{r_2}{r_{\alpha}(Q)} \ \mathrm{min}(r_{\alpha}(P),r_{\alpha}(Q))$ \ $(\alpha  \in L)$}
\\
	\bottomrule 
\end{tabular}
\end{center}
\caption{Operational semantics for PEPA components}
\label{table:rules}
\end{table*}

\section{The Calculus}\label{sec:calculus}
PEPA (Performance Evaluation Process Algebra)  \cite{hillston:book} is an
algebraic calculus  enhanced with stochastic timing infor\-mation
 which may~be used to~calculate performance measures 
as well as prove functional system~properties. 
%
%
%
%

The basic elements of PEPA are \emph{components} and \emph{activities}.
Each activity is represented by a pair $(\alpha, r)$ where $\alpha$ is a label, or \emph{action type}, 
and $r$ is its  \emph{activity rate},  that is the parameter of a negative
exponential distribution determining its duration. 
We assume that
there is a countable
 set, $\cA$,  of possible action types, including a distinguished
type, $\tau$,
which can be regarded as the \emph{unknown} type.
Activity rates may be any positive real number, or the distinguished symbol
$\top$ which should be read as \emph{unspecified}.
%
%

The syntax for  PEPA terms is defined by the  grammar:
$$\begin{array}{cclccl}
P & ::= & P \sync{L} P \mid P/L \mid S\\[1mm]
S & ::= & (\alpha, r).S \mid S+S \mid A
\end{array}$$
where $S$ denotes a \emph{sequential component}, while $P$ denotes a
\emph{model component} which executes in parallel. 
We assume that there is  a countable set of \emph{constants},
$A$.
We write $\cC$ for the set of all possible
components.

\subsection{Structural Operational Semantics}
PEPA is given a structural operational semantics, as shown in Table \ref{table:rules}.
The component $(\alpha, r).P$ carries out the activity $(\alpha, r)$
of type  $\alpha$ at rate $r$ and subsequently behaves as  $P$. When $a=(\alpha, r)$, the component $(\alpha, r).P$  may be written as $a.P$.
The component $P+Q$ represents a system which may behave either as  $P$ or as  $Q$. $P+Q$ enables all the current activities of both $P$ and $Q$. The first activity to complete distinguishes one of the components, $P$ or $Q$. The other component of the choice is discarded. 
%
The component $P/L$ behaves as $P$ except that any activity of type within the set $L$ are \emph{hidden}, i.e., they are relabeled with the unobservable type $\tau$. The meaning of a constant $A$ is given by a defining equation such as 
$A\rmdef P$ which gives the constant $A$ the behavior of the component~$P$. 
The cooperation combinator $\sync{L}$ is in fact an indexed family of combinators, one
for each possible set of action types, $L\subseteq \cA\setminus \{\tau\}$.
The \emph{cooperation set} $L$  defines the action types on which the components must synchronize or \emph{cooperate} (the unknown action type, $\tau$, may not appear in any cooperation set).
It is assumed that each component proceeds independently with any activities whose types do not occur
in the cooperation set $L$ (\emph{individual activities}). However, activities with action types in the set $L$ require
the simultaneous involvement of both components (\emph{shared activities}). These shared activities will only
be enabled in $P\sync{L}Q$ when they are enabled in both $P$ and $Q$. 
The shared activity will have the same action type as the two contributing activities and a rate
reflecting the rate of the slower participant \cite{hillston:book}.
If an activity has an unspecified rate in a component, the
component is passive~with respect to that action type. In this case 
the rate of 
the shared activity will be completely determined by the
other component.
For a given  $P$ and action  type $\alpha$, 
this is the \emph{apparent rate} \cite{marin:valuetools13} of 
$\alpha$ in $P$,
denoted $r_{\alpha}(P)$, that is the sum of the rates of the 
$\alpha$  activities enabled in~$P$. 

The semantics of each
term in PEPA is given via a labeled \emph{multi-transition system} where the 
multiplicities of arcs are significant. In the transition system, a
state or
\emph{derivative} corresponds to each syntactic term of the language 
and an arc represents the activity which causes
one derivative to evolve into another. The  set of
reachable states of a model $P$ is termed the \emph{derivative set} of $P$, denoted by $ds(P)$,
 and
constitutes the set of nodes of the \emph{derivation graph} of $P$ ($\cD(P)$) obtained
by applying the semantic rules exhaustively.
%
%
%
We denote by $\cA(P)$ the set of all the \emph{current action types} of $P$, i.e.,
 the set of action types which the component $P$ may next engage in.
We denote by $\Ac(P)$ the multiset of all the \emph{current activities} of $P$. 
%
%
%
%
Finally we denote by ${\vec{\cal A}(P)}$ the union of all ${\cA{(P')}}$ with $P'\in ds(P)$, i.e., the set of all action types syntactically occurring in $P$.
For any component $P$, the \emph{exit rate} from $P$ will be the sum of the 
activity rates of all the activities
enabled in $P$, i.e., $ q(P) = \sum_{a \in \Ac(P)}r_a$, with
$r_a$ being the rate of activity $a$. 
If $P$ enables more than one activity, $|\Ac(P)|>1$, then the
dynamic behavior of the model is determined by a race condition. 
This has the effect of replacing the
nondeterministic branching of the pure process algebra with probabilistic 
branching. The probability
that a particular activity completes is given by the ratio of the activity 
rate to the exit rate from $P$.

\subsection{Underlying Stochastic Process}
In \cite{hillston:book} it is proved that
for any finite PEPA model $P\rmdef P_0$ \com{with $ds(P)=\{P_0,\ldots,P_n\}$}, if we define the stochastic process 
$X(t)$, such that $X(t)=P_i$ indicates that the system behaves as component 
$P_i$ at time t, then $X(t)$ is a continuous time
Markov chain.

The \emph{transition rate} between two components $P_i$ and $P_j$, denoted 
 $q(P_i,P_j)$,~is the rate at
which the system changes from behaving as component $P_i$ to behaving as $P_j$. 
It is the sum of the
activity rates labeling arcs which connect the node cor\-responding to $P_i$ 
to the node corresponding to $P_j$ in $\cD(P)$,~i.e., 
\begin{center}
$ q(P_i,P_j)= \sum_{a\in \Ac(P_i| P_j)} r_a$
\end{center}
where $P_i\not = P_j$ and $\Ac(P_i| P_j) = \bms a\in \Ac(P_i) |\ P_i \transits{a} P_j\ems$. 
Clearly if $P_j$ is not a one-step derivative of $P_i$, $q(P_i,P_j)=0$.
The $q(P_i,P_j)$ (also denoted $q_{ij}$), are the off-diagonal elements of the 
infinitesimal generator matrix of the
Markov process, ${\bf Q}$. Diagonal elements are formed as the negative sum of 
the non-diagonal elements of
each row. We use the following notation: $q(P_i)=\sum_{j\neq i}q(P_i,P_j)$ and
 $q_{ii}=-q(P_i)$. 
For any finite and irreducible PEPA model $P$, the steady-state distribution 
$\mathrm{\Pi}(\cdot)$ exists and it may be found by solving the normalization equation and the global
balance equations: $\sum_{P_i\in ds(P)}\mathrm{\Pi}(P_i)=1$ and $\mathrm{\Pi} \mathbf{Q} =\mathbf{0}$.
%
%
%
The \emph{conditional transition rate} from $P_i$ to $P_j$ via an action type 
$\alpha$ is denoted $q(P_i,P_j,\alpha)$. This is
the sum of the activity rates labeling arcs connecting the corresponding 
nodes in the derivation graph
which are also labeled by the action type $\alpha$.
 It is the rate at which a system behaving as component $P_i$
evolves to behaving as component $P_j$ as the result of completing a 
type $\alpha$ activity.
%
The \emph{total conditional transition rate} from $P$ to $S\subseteq ds(P)$,
denoted $q[P,S,\alpha]$, is defined as
\[
q[P,S,\alpha]=\sum_{P'\in S} q(P,P',\alpha)
\]
where $q(P,P',\alpha)=\sum_{P \xrightarrow{(\alpha,r_{\alpha})} P'} r_{\alpha}$.

\subsection{Observation Equivalence}
In a process algebra, actions, rather than states, play the role of
capturing the observable behavior of a system  model. This leads to a formally defined notion 
of equivalence in which components are regarded
as equal if, under observation, they appear to perform exactly the same 
actions. In this section we recall a bisimulation-like relation,
named \emph{lumpable bisimilarity},
 for PEPA models \cite{marin:valuetools13}.



Two PEPA components are \emph{lumpably bisimilar}
 if there is an equivalence relation between them such that,
for any action type $\alpha$ different from $\tau$, the total conditional transition rates from those components to any equivalence
class, via activities of this type, are the same.

\begin{definition}\emph{(Lumpable bisimulation)}
\label{def:bisimulation}
An equivalence relation over PEPA components, $\cR\subseteq \cC\times \cC$, is a \emph{lumpable bisimulation} if whenever $(P,Q)\in \cR$ then for  all $\alpha\in \cA$ and for all  $S\in \cC/\cR$ such that
\begin{itemize}
\item either $\alpha\not = \tau$,
\item   or $\alpha=\tau$ and $P,Q\not \in S$,
\end{itemize}
it holds
$$q[P,S,\alpha]=q[Q,S,\alpha]\, .$$
\end{definition}


It is clear that the identity relation is a lumpable bisimulation.
We are interested in the relation which is the largest lumpable bisimulation, 
formed by the union of all lumpable bisimulations. 

%


\begin{definition}\emph{(Lumpable bisimilarity)}
\label{def:bisimilarity}
Two PEPA components $P$ and $Q$ are \emph{lumpably bisimilar}, 
written 
$P\approx_lQ$, if $(P,Q)\in\cR$ for some lumpable bisimulation $\cR$,~i.e.,
$$\approx_l \ =\bigcup \ \{\cR\ |\ \cR \mbox{ is a lumpable bisimulation}\}.$$
\noindent
$\approx_l$ is called \emph{lumpable bisimilarity} and it is the largest symmetric lumpable bisimulation over PEPA components.
\end{definition}

In \cite{marin:valuetools13} we proved that lumpable bisimilarity is a congruence for the so-called evaluation contexts, i.e.,
if $P_1 \approx_l P_2$ then 
\begin{itemize}
\item $a.P_1\approx_l a.P_2$;
\item 
 $P_1\sync{L}Q\approx_l P_2\sync{L}Q\ \ $ for all $L\subseteq \cA$.
\item  $P_1/L\approx_lP_2/L$. 
\end{itemize}

Notice that the notion of strong equivalence defined in  \cite{hillston:book} is stricter
  than that of lumpable bisimilarity because the latter allows arbitrary activities with
  type $\tau$ among components belonging to the same equivalence class. 

In \cite{hermanns:2005} a notion of weak bisimulation for CTMCs is introduced.
This  is based on the idea that the time-abstract behavior of equivalent states is weakly bisimilar and that the relative speed of these states to move to a different equivalence class is equal. To capture this intuition, the authors propose a definition of weak-bisimulation which resembles our notion  of lumpable bisimulation if we ignore action types and labels.
This bisimulation is defined in the context of both discrete and continuous time Markov chains 
without any notion of compositionality, and hence of contextuality. 
Compositionality is considered in \cite{AB09,boudali:2007,hennessy:2013}, where
definitions of  weak bisimilarities for stochastic process algebra based on the classical concept of weak action  are proposed.
Our approach shares with these bisimilarities the idea of ignoring the rates for non-synchronizing (labeled $\tau$) transitions between a state and the others belonging to the same equivalence class.
The main difference between our definition and those presented in \cite{AB09,boudali:2007,hennessy:2013} is that 
we explicitly studied  the relationships between 
our lumpable bisimilarity at the process algebra level and the induced lumping of the underlying Markov chains. This led to a coinductive characterization of a notion of contextual lumpability as described in \cite{marin:valuetools13}.

\section{Persistent Stochastic Non-Interference}\label{sec:psni}
The security property named \emph{Persistent Stochastic Non-Interference (PSNI)} tries to capture every possible information flow from a \emph{classified (high)} level of confidentiality to an \emph{untrusted (low)} one. A strong requirement of this definition is that no information flow should be possible even in the presence of malicious processes that run at the classified level. The main motivation is to protect a system also from internal attacks, which could be performed by the so-called \emph{Trojan Horse} programs, i.e., programs that appear honest but hide some malicious code inside them.


More precisely, the notion of \emph{PSNI} consists of checking   all the states reachable 
by the system against all high level potential interactions.

In order to formally define our security property, we partition the set $\cA\setminus \{\tau\}$  of visible action types, into two sets, ${\cal H}$ and ${\cal L}$ of high and low level action types. A high level PEPA component $H$ is a PEPA term such that for all $H'\in  ds(H)$,  $\cA(H')\subseteq {\cal H}$, i.e., every derivative of $H$ may next engage in only high level actions. We denote by ${\cal C}_H$  the set of all high level PEPA components.

A system $P$ satisfies \emph{PSNI} if for every state $P'$ reachable from $P$ and for every high level process $H$ a low level user cannot distinguish $P'$ from $P'\sync{{\cal H}}H$. In other words, a system $P$ satisfies \emph{PSNI} if what a low level user sees of the system is not modified when it cooperates with any high level process $H$.

In order to formally define the $\emph{PSNI}$ property,  we denote by  $P\setminus {\cal H}$ the PEPA component $(P\sync{{\cal H}}\bar{H})$ where $\bar{H}$ is any high level process that does not cooperate with $P$, i.e., 
for all $P'\in ds(P)$, ${\cal A}(P')\cap {\cal A}(\bar{H})=\emptyset$. Intuitively $P\setminus {\cal H}$ denotes the component $P$ prevented from performing high level actions.
Notice that the definition is well formed  in the sense that if $\bar{H_1}$ and $\bar{H_2}$ are two high level processes that do not cooperate with $P$, then the derivation graphs of $(P\sync{{\cal H}}\bar{H_1})$ and $(P\sync{{\cal H}}\bar{H_2})$ are isomorphic.
 
Properties  $\emph{SNI}$ and $\emph{PSNI}$ are formally defined as follows.

\begin{definition}\emph{(Stochastic Non-Interference)}
Let $P$ be a PEPA component. 
\[P\in SNI \mbox{ iff }
\forall H\in {\cal C}_H,\]
\[ P\setminus {\cal H}\approx_l (P\sync{{\cal H}}H)/{\cal H}\,.
\]
\end{definition}

\begin{definition}\emph{(Persistent Stochastic Non-Interference)}
\label{def:sni}
Let $P$ be a PEPA component. 
\[P\in PSNI \mbox{ iff }
\forall P'\in ds(P), \, \forall H\in {\cal C}_H,\]
\[ P'\in SNI, \mbox{ i.e., } P'\setminus {\cal H}\approx_l (P'\sync{{\cal H}}H)/{\cal H}\,.
\]
\end{definition}

We introduce a novel bisimulation-based equivalence relation over PEPA components, named $\approx_l^{hc}$, that allows us to give a first characterization of \emph{PSNI} with no quantification over all the high level components $H$. In particular, we show that $P\in \mathit{PSNI}$ if and only if $P\setminus {\cal H}$ and $P$ are not distinguishable with respect to $\approx_l^{hc}$. Intuitively, two processes are 
$\approx_l^{hc}$-equivalent if they can simulate each other in any possible high context, i.e., in every context $C[\_]$ of the form $(\_ \ \sync{{\cal H}}H)/{\cal H}$ where $H\in {\cal C}_H$.
Observe that for any high context $C[\_]$ and PEPA model $P$, all the states reachable from $C[P]$ have the form $C'[P']$ with $C'[\_]$ being a high context too and $P'\in ds(P)$.

We now introduce the concept of \emph{lumpable bisimulation on high contexts}: the idea is that, given two PEPA models $P$ and $Q$, when a high level context $C[\_]$ filled with $P$ executes a certain activity moving $P$ to $P'$ then the same context filled with $Q$ is able to simulate this step moving $Q$ to $Q'$ so that $P'$ and $Q'$ are again lumpable bisimilar on high contexts, and vice-versa. This must be true for every possible high context $C[\_]$. It is important to note that the quantification over all possible high contexts is re-iterated for $P'$ and $Q'$.
%
For a PEPA model $P$,  $\alpha\in \cA$,  $S\subseteq ds(P)$ and a high context $C[\_]$
we define:

\[
 q_C(P,P',\alpha)=\sum_{C[P] \xrightarrow{(\alpha,r_{\alpha})} C'[P']} r_{\alpha}
\]
and
\[
q_C[P,S,\alpha]=\sum_{P'\in S} q_C(P,P',\alpha)\,.
\]

The notion of \emph{lumpable bisimulation on high contexts} is defined as follows:

\begin{definition}\emph{(Lumpable bisimilarity on high contexts)}
\label{def:hcontexts}
An equivalence relation over PEPA components, $\cR\subseteq \cC\times \cC$, is a \emph{lumpable bisimulation on high contexts} if whenever $(P,Q)\in \cR$ then for all high context $C[\_]$,
for  all $\alpha\in \cA$ and for all  $S\in \cC/\cR$ such that
\begin{itemize}
\item either $\alpha\not = \tau$,
\item   or $\alpha=\tau$ and $P,Q\not \in S$,
\end{itemize}
it holds
$$q_C[P,S,\alpha]=q_C[Q,S,\alpha]\, .$$
Two PEPA components $P$ and $Q$ are \emph{lumpably bisimilar on high contexts}, 
written 
$P\approx^{hc}_lQ$, if $(P,Q)\in\cR$ for some lumpable bisimulation on high contexts $\cR$,~i.e.,
$$\approx^{hc}_l \ =\bigcup \ \{\cR\ |\ \cR \mbox{ is a lumpable bisimulation on high contexts}\}.$$
\noindent
$\approx^{hc}_l$ is called \emph{lumpable bisimilarity on high contexts} and it is the largest symmetric lumpable bisimulation on high contexts  over PEPA components. It is easy to prove that $\approx^{hc}_l$  is an equivalence relation.
\end{definition}

The next theorem gives a characterization of \emph{PSNI} in terms of  $\approx^{hc}_l$.

\begin{theorem}
Let $P$ be a PEPA component. Then \[P\in PSNI \mbox{ iff } P\setminus{\cal H}\approx^{hc}_l P\,.\]
\end{theorem}
\begin{proof}
We first show that $P\setminus{\cal H}\approx^{hc}_l P$ implies $P\in PSNI$. In order to do it we prove that
\[{\cal R}=\{(P_1\setminus{\cal H}, (P_2\sync{{\cal H}}H)/{\cal H})\,|\, H\in {\cal C}_H \mbox{ and } P_1\setminus{\cal H}\approx^{hc}_l P_2\}\]
is a lumpable bisimulation. This is sufficient to say that $P\in PSNI$.

First observe that, 
if $P\setminus{\cal H}\approx^{hc}_l P$ then for all 
 $P'\in ds(P)$ there exists
$P''\setminus {\cal H}\in ds(P\setminus {\cal H})$  such that
$P''\setminus {\cal H}\approx^{hc}_l P'$ and, 
by definition of ${\cal R}$,  for all 
$H\in {\cal C}_H$, $(P''\setminus {\cal H}, (P'\sync{{\cal H}}H)/{\cal H})\in {\cal R}$. Since ${\cal R}$ is a lumpable bisimulation, we have that for all 
$H\in {\cal C}_H$, $P''\setminus {\cal H}\approx_l (P'\sync{{\cal H}}H)/{\cal H}$. In particular, there exists $\bar{H}\in {\cal C}_H$ such that 
$(P'\sync{{\cal H}}\bar{H})/{\cal H}$ coincides with 
 $P'\setminus{\cal H}$.
Since $\approx_l $ is an equivalence relation, by symmetry and transitivity, we have that for every $P'\in ds(P)$ and for every 
$H\in {\cal C}_H$, $P''\setminus {\cal H}\approx_l P'\setminus {\cal H} \approx_l (P'\sync{{\cal H}}H)/{\cal H}$, i.e., $P\in PSNI$.
The fact that ${\cal R}$ is a lumpable bisimulation follows from: 
\begin{itemize}
\item  if $P_1\setminus {\cal H}\approx^{hc}_l P_2$ then for all $\alpha\in\cA$ with $\alpha \neq \tau$ and for all $S\in \cC/\!\approx^{hc}_l$ and for all high context $C[\_]$, we have $q_C[P_1\setminus{\cal H}, S,\alpha]=q_C[P_2, S,\alpha]$.   
Since a high context can only perform high level activities, we have that for all high level context $C[\_]$, it holds that $q[P_1\setminus{\cal H}, S,\alpha]=q_C[P_1\setminus{\cal H}, S,\alpha]$ and then $q[P_1\setminus {\cal H}, S,\alpha]=q_C[P_2, S,\alpha]$, i.e., we have  that for all $(P_1\setminus {\cal H}, (P_2\sync{{\cal H}}H)/{\cal H})\in {\cal R}$ and for all   $S'\in \cC/{\cal R}$ 
it holds
$q[P_1\setminus {\cal H}, S',\alpha]=q[(P_2\sync{{\cal H}}H)/{\cal H}, S',\alpha]$.
\item if $P_1\setminus {\cal H}\approx^{hc}_l P_2$ then for   $\alpha = \tau$ and for all $S\in \cC/\!\approx^{hc}_l$ with $P_1\setminus {\cal H}, P_2\not \in S$ and for all high context $C[\_]$, we have $q_C[P_1\setminus{\cal H}, S,\alpha]=q_C[P_2, S,\alpha]$.   Since a high context can only perform high level activities, we have that for all high level context $C[\_]$, it holds that $q[P_1\setminus{\cal H}, S,\alpha]=q_C[P_1\setminus{\cal H}, S,\alpha]$. Hence  for all $(P_1\setminus {\cal H}, (P_2\sync{{\cal H}}H)/{\cal H})\in {\cal R}$ and for all   $S'\in \cC/{\cal R}$ with $P_1\setminus {\cal H},(P_2\sync{{\cal H}}H)/{\cal H})\not \in S'$
it holds
$q[P_1\setminus {\cal H}, S',\alpha]=q[(P_2\sync{{\cal H}}H)/{\cal H}, S',\alpha]$.
\end{itemize}
We now show that if $P\in PSNI$ then $P\setminus {\cal H}\approx^{hc}_l P$. To this end it is sufficient to prove that
\[{\cal R}=\{(P_1\setminus {\cal H}, P_2)\,|\,  P_1\setminus{\cal H}\approx_l P_2\setminus{\cal H} \mbox{ and } P_2\in PSNI\}\]
is a lumpable bisimulation on high contexts. Indeed, let $C[\_]$ be a high context and $\alpha\in\cA$.
\begin{itemize}
\item Assume $\alpha \neq \tau$.  From  $P_1\setminus {\cal H}\approx_l P_2\setminus {\cal H}$, we have that for all $S\in \cC/\approx_l$, 
$q[P_1\setminus {\cal H}, S,\alpha]=q[P_2\setminus {\cal H}, S,\alpha]$.
Since a high context can only perform high level activities, we have that for all high context $C[\_]$ it holds  $q[P_1\setminus {\cal H}, S,\alpha]=q_C[P_1\setminus {\cal H}, S,\alpha]$. Moreover, since $\alpha\neq \tau$, $q[P_2\setminus {\cal H}, S,\alpha]=q_C[P_2, S',\alpha]$ where $S'=\{P\,|\, P\setminus {\cal H}\in S\}$,
i.e., for all high context $C[\_]$ and $S\in \cC/{\cal R}$ it holds
$q_C[P_1\setminus {\cal H}, S,\alpha]=q_C[P_2, S,\alpha] $.
\item
Consider now $\alpha=\tau$. From $P_1\setminus{\cal H}\approx_l P_2\setminus{\cal H}$, we have that for all $S\in \cC/\approx_l$ such that $P_1\setminus {\cal H}$, $P_2\setminus {\cal H}\not \in S$, 
$q[P_1\setminus {\cal H}, S,\alpha]=q[P_2\setminus {\cal H}, S,\alpha]$. Since a high context can only perform high level activities and both $P_1\setminus{\cal H}$ and $P_2\setminus {\cal H}$ do not perform high activities, we have that  $q[P_i\setminus {\cal H}, S,\alpha]=q_C[P_i\setminus{\cal H}, S,\alpha]$ for all high level context $C[\_]$ and for $i\in \{1,2\}$. From the fact that $P_2\in PSNI$, we have $P_2\setminus{\cal H}\approx_l (P_2\sync{{\cal H}}H)/{\cal H}$ for all
$H\in {\cal C}_H$, and then 
$q[P_2\setminus {\cal H}, S,\alpha]=q_C[P_2\setminus {\cal H}, S,\alpha]=q_C[P_2, S',\alpha]$ for all high context $C[\_]$, 
$S\in \cC/\approx^{hc}_l$ and $S'\in \cC/{\cal R}$ such that $P_2\setminus {\cal H}\not \in S$ and $P_2\not \in S'$, i.e.,
$q_C[P_1\setminus {\cal H}, S,\alpha]=q_C[P_2, S,\alpha]$ for all high context $C[\_]$ and 
$S\in \cC/{\cal R}$ such that $P_1\setminus{\cal H},P_2\not \in S$.
\end{itemize}
\end{proof}

Finally, we show how it is possible to give a characterization of \emph{PSNI} avoiding both the universal quantification over all the possible high level components and the universal quantification over all the possible reachable states.

Before we have shown how the idea of ``being secure in every state'' can be directly moved  inside the lumpable bisimulation on high contexts notion ($\approx^{hc}_l$). However this bisimulation notion implicitly contains a quantification over all possible high contexts. We now prove that $\approx^{hc}_l$ can be expressed in a rather simpler way by exploiting local information only. This can be done by defining a novel equivalence relation which focuses only on observable actions that do not belong to ${\cal H}$. More in detail, we define an observation equivalence where actions from ${\cal H}$ \emph{may} be ignored.

We first introduce the notion of \emph{lumpable bisimilarity up to ${\cal H}$}.

\begin{definition}\emph{(Lumpable bisimilarity up to ${\cal H}$)}
\label{def:upto}
An equivalence relation over PEPA components, $\cR\subseteq \cC\times \cC$, is a \emph{lumpable bisimulation up to ${\cal H}$} if whenever $(P,Q)\in \cR$ then 
for  all $\alpha\in \cA$ and for all  $S\in \cC/\cR$ 
\begin{itemize}
\item if $\alpha \not \in {\cal H}\cup\{\tau\}$ then $$q[P,S,\alpha]=q[Q,S,\alpha]\, ,$$
\item if $\alpha \in {\cal H}\cup\{\tau\}$ and $P,Q\not \in S$, then $$q[P,S,\alpha]=q[Q,S,\alpha]\,.$$
\end{itemize}
Two PEPA components $P$ and $Q$ are \emph{lumpably bisimilar up to ${\cal H}$}, 
written 
$P\approx^{\cal H}_lQ$, if $(P,Q)\in\cR$ for some lumpable bisimulation up to ${\cal H}$,~i.e.,
$$\approx^{\cal H}_l \ =\bigcup \ \{\cR\ |\ \cR \mbox{ is a lumpable bisimulation up to \mbox{${\cal H}$}}\}.$$
\noindent
$\approx^{\cal H}_l$ is called \emph{lumpable bisimilarity  up to ${\cal H}$} and it is the largest symmetric lumpable bisimulation  up to ${\cal H}$ over PEPA components.
\end{definition}

The next theorem shows that the binary relations $\approx^{hc}_l$ and 
$\approx^{{\cal H}}_l$ are equivalent.

\begin{theorem}\label{teo:H-hc}
Let $P$ and $Q$ be two PEPA components. Then \[P \approx^{hc}_l \, Q \mbox{ if and only if  }\,
P\approx^{\cal H}_l Q\,.\]
\end{theorem}
\begin{proof}
We first show that $P\approx^{hc}_l Q$ implies $P\approx^{\cal H}_l Q$. In order to do it we prove that
\[{\cal R}=\{(P,Q)\,|\, P\approx^{hc}_l Q\}\]
is a lumpable bisimulation up to ${\cal H}$. This follows from the following cases. First observe that, by definition of ${\cal R}$, $S\in \cC/\!\approx^{hc}_l$ if and only if $S\in  \cC/\!{\cal R}$.
 
\begin{itemize}
\item  Let $\alpha \not \in {\cal H}\cup\{\tau\}$. From the fact that 
$P\approx^{hc}_l Q$ it holds that  for all $S\in \cC/\!\approx^{hc}_l$ and for all high context $C[\_]$,  $q_C[P, S,\alpha]=q_C[Q, S,\alpha]$. Since 
  $\alpha \not \in {\cal H}\cup\{\tau\}$, we have that $q[P, S,\alpha]=q[Q, S,\alpha]$.
  \item  Let $\alpha \in {\cal H}\cup\{\tau\}$. From the fact that 
    $P\approx^{hc}_l Q$ it holds that  for all $S\in \cC/\!\approx^{hc}_l$ such that $P,Q\not \in S$ and for all high context $C[\_]$,  $q_C[P, S,\tau]=q_C[Q, S,\tau]$. If $C[\_]$ does not synchronize with $P$, 
 we have that $q[P, S,\tau]=q[Q, S,\tau]$. On the other hand, consider a context $C[\_]$ with only one current action type $h\in {\cal H}$.
 Then, from $q_C[P, S,\tau]=q_C[Q, S,\tau]$ and  $q[P, S,\tau]=q[Q, S,\tau]$, it follows that if $P$ cooperates over $h$ then also $Q$ cooperates over $h$ and $q[P, S,h]=q[Q, S,h]$. 
\end{itemize}

We now show that if $P\approx^{\cal H}_l Q$ then  $P\approx^{hc}_l Q$. To this end it is sufficient to prove that
\[{\cal R}=\{(P,Q)\,|\, P\approx^{\cal H}_l Q\}\]
is a lumpable bisimulation on high contexts. This follows from the following cases. First observe that, by definition of ${\cal R}$, $S\in \cC/\!\approx^{hc}_l$ if and only if $S\in  \cC/\!{\cal R}$.

\begin{itemize}
\item  Let $\alpha \not \in {\cal H}\cup\{\tau\}$. From the fact that 
$P\approx^{\cal H}_l Q$ it holds that for all $S\in \cC/\!\approx^{\cal H}_l$,
$q[P, S,\alpha]=q[Q, S,\alpha]$.
Since a high context can only perform high level activities, we have that  $q[P, S,\alpha]=q_C[P, S,\alpha]$ and $q[Q, S,\alpha]=q_C[Q, S,\alpha]$ for all high context $C[\_]$. Hence, 
$q_C[P, S,\alpha]=q_C[Q, S,\alpha]$.
\item Let $\alpha =\tau$. From the fact that 
$P\approx^{\cal H}_l Q$ it holds that for all $S\in \cC/\!\approx^{\cal H}_l$ such that $P,Q\not \in S$, 
$q[P, S,\alpha]=q[Q, S,\alpha]$. Hence for all high level context that do not synchronize with $P$ and $Q$ we have that
 $q[P, S,\alpha]=q_C[P, S,\alpha]$ and $q[Q, S,\alpha]=q_C[Q, S,\alpha]$, i.e.,
$q_C[P, S,\alpha]=q_C[Q, S,\alpha]$.
\item Let $h\in{\cal H}$.  From the fact that 
$P\approx^{\cal H}_l Q$ it holds that for all $S\in \cC/\!\approx^{\cal H}_l$ such that $P,Q\not \in S$, 
  $q[P, S,h]=q[Q, S,h]$. From this and the fact that $q[P, S,\tau]=q[Q, S,\tau]$ it follows that
 for all high level context $C[\_]$ with only one current action type $h\in {\cal H}$, 
  $q_C[P, S,\tau]=q_C[Q, S,\tau]$.
By induction on the number of current action types of a high level context $C[\_]$, we obtain that for $\alpha=\tau$, for all $S\in \cC/{\cal R}$ 
with $P,Q\not \in S$ it holds $q_C[P, S,\alpha]=q_C[Q, S,\alpha]$.

\end{itemize}
\end{proof}

Theorem \ref{teo:H-hc} allows us to identify a local property of processes (with no quantification on the states and on the high contexts) which is a necessary and sufficient condition for  \emph{PSNI}. This is stated by the following corollary:

\begin{corollary}
  \label{cor:psni}
Let $P$ be a PEPA component. Then \[P\in PSNI \mbox{ iff } P\setminus {\cal H}\approx^{\cal H}_l P\,.\]
\end{corollary}

Finally we provide a characterization of \emph{PSNI} in terms of \emph{unwinding conditions}. In practice, whenever a state $P'$ of a \emph{PSNI} PEPA model $P$ may execute a high level activity leading it to a state $P''$, then $P'$ and $P''$ are indistinguishable for a low level observer.

\begin{theorem}\label{theo:unwinding}
Let $P$ be a PEPA component.
\[P\in PSNI \mbox{ iff }
\forall P'\in ds(P), \, \] \[P' \transits{(h,r)} P'' \mbox{ implies }
P' \setminus {\cal H}\approx_l P''\setminus{\cal H}\]
\end{theorem}
\begin{proof}
  We first prove that if $P\in PSNI $ then for all $P'\in ds(P)$, $P' \transits{(h,r)}~P''$ implies $P' \setminus {\cal H}\approx_l P''\setminus{\cal H}$. Indeed, by Definition \ref{def:sni}, $P'\in PSNI$ and therefore, by
  Corollary \ref{cor:psni}, $ P'\setminus {\cal H}\approx^{\cal H}_l P'$. By Definition \ref{def:upto} of $\approx^{\cal H}_l$, for all $S\in {\cal C}/\approx^{\cal H}_l$ such that $P'\setminus {\cal H}, P'\not \in S$, both
$q[P'\setminus {\cal H},S,\tau]=q[P',S,\tau]$ and
  $q[P'\setminus {\cal H},S,h]=q[P',S,h]$. Since $P'\setminus {\cal H}$ does not perform any high level action, $q[P'\setminus {\cal H},S,h]=0$ while, since $P' \transits{(h,r)} P''$,
  $q[P',S,\hat{h}]\neq 0$. Therefore, from $ P'\setminus {\cal H}\approx^{\cal H}_l P'$,
either $h$ is not a current action type of $P'$ or $P'\setminus {\cal H}, P' \in S$, i.e.,  $P'\setminus {\cal H}\approx^{\cal H}_l P''$. Since also $P''\in PSNI$, from $P''\setminus {\cal H}\approx^{\cal H}_l P''$ it follows that $P'\setminus {\cal H}\approx^{\cal H}_l P''\setminus {\cal H}$. Finally, since both $P'\setminus {\cal H}$ and $P''\setminus {\cal H}$ do not perform any high level activity, $P'\setminus {\cal H}\approx^{\cal H}_l P''\setminus {\cal H}$ is equivalent to $P'\setminus {\cal H}\approx_l P''\setminus {\cal H}$.

  We now prove that if  for all $P'\in ds(P)$, $P' \transits{(h,r)}~P''$ implies $P' \setminus {\cal H}\approx_l P''\setminus{\cal H}$ then $P\in PSNI$. Indeed observe that for all $\alpha\not \in {\cal H}\cup{\tau}$, and for all $S\in {\cal C}/\approx^{\cal H}_l$, $q[P'\setminus {\cal H},S,\alpha]=q[P',S,\alpha]$. Moreover, if $P'\setminus {\cal H}, P'\not \in S$ then $q[P'\setminus {\cal H},S,\tau]=q[P',S,\tau]$. This is sufficient to prove that $P'\setminus {\cal H}\approx^{\cal H}_l P'$, i.e., by Corollary \ref{cor:psni}, $P\in PSNI$.
\end{proof}

\section{Properties of Persistent Stochastic Non-Interference}\label{sec:compositionality}
In this section we prove some interesting propertis of $PSNI$. 
First we prove that $PSNI$ is compositional with respect to  low prefix, cooperation over low actions and hiding.

\begin{proposition} Let $P$ and $Q$ be two PEPA components. If $P,Q\in PSNI$, then
\begin{itemize}
\item $(\alpha,r).P\in PSNI\,$ for all $\alpha\in {\cal L}\cup\{\tau\}$
\item $P/L\in PSNI\,$ for all $L\subseteq {\cal A}$
\item $P\sync{L}Q\in PSNI\,$  for all $L\subseteq {\cal L}$
\end{itemize}
\end{proposition}
\begin{proof}
Assume that $P,Q\in PSNI$.
\begin{itemize}
\item If $P\in PSNI$ then for all $P'\in ds(P)$, $P' \transits{(h,r)} P''$ implies   $P' \setminus {\cal H}\approx_l P''\setminus{\cal H}$. This property is clearly maintained for the PEPA component $(\alpha,r).P$ when $\alpha\in {\cal L}\cup\{\tau\}$.
\item If $P\in PSNI$ then for all $P'\in ds(P)$, $P' \transits{(h,r)} P''$ implies   $P' \setminus {\cal H}\approx_l P''\setminus{\cal H}$. Let  $L\subseteq {\cal A}$ and $P'/L\in ds(P)$. Assume that $P'/L \transits{(h,r)} P''/L$. From the fact that $P' \setminus {\cal H}\approx_l P''\setminus{\cal H}$ we have that
$(P'\sync{{\cal H}}\bar{H})\approx_l(P''\sync{{\cal H}}\bar{H})$  for any high level PEPA component $\bar{H}$ that does not cooperate with $P$. From the fact that lumpable bisimilarity is a congruence for the evaluation contexts, we have that for all $L\subseteq {\cal A}$, $(P'\sync{{\cal H}}\bar{H})/L\approx_l(P''\sync{{\cal H}}\bar{H})/L$. We can assume that $\vec{\cal A}(\bar{H})\cap L=\emptyset$ and hence, since also $\vec{\cal A}(\bar{H})\cap \vec{\cal A}(\bar{P})=\emptyset$,  $(P'/L\sync{{\cal H}}\bar{H})/L\approx_l(P''/L\sync{{\cal H}}\bar{H})/L$, i.e.,  $(P'/L) \setminus {\cal H}\approx_l (P''/L)\setminus{\cal H}$.
\item If $P,Q\in PSNI$ then for all $P'\in ds(P)$, $P' \transits{(h,r)} P''$ implies   $P' \setminus {\cal H}\approx_l P''\setminus{\cal H}$ and for all $Q'\in ds(Q)$, $Q' \transits{(h,r)} Q''$ implies   $Q' \setminus {\cal H}\approx_l Q''\setminus{\cal H}$. Let  $L\subseteq {\cal L}$ and $P'\sync{L}Q'\in ds(P\sync{L}Q)$. Assume that $P'\sync{L}Q'\transits{(h,r)} P''\sync{L}Q''$.  In this case, either $P'\transits{(h,r)} P''$ or 
$Q'\transits{(h,r)} Q''$. Assume  that $P'\transits{(h,r)} P''$ and then
$P'\sync{L}Q'\transits{(h,r)} P''\sync{L}Q'$.
From the hypothesis that $P\in PSNI$ we have that $P' \setminus {\cal H}\approx_l P''\setminus{\cal H}$, i.e., $(P'\sync{{\cal H}}\bar{H})\approx_l(P''\sync{{\cal H}}\bar{H})$  for any high level PEPA component $\bar{H}$ that does not cooperate with $P$ and $Q$.
From the fact that $\approx_l$ is a congruence with respect to the cooperation operator we have $(P'\sync{{\cal H}}\bar{H})\sync{L}(Q'\sync{{\cal H}}\bar{H})\approx_l(P''\sync{{\cal H}}\bar{H})\sync{L}(Q'\sync{{\cal H}}\bar{H})$, moreover sice ${\cal H}\cap L=\emptyset$ we obtain
$(P'\sync{L}Q')\sync{{\cal H}}\bar{H} \approx_l (P''\sync{L}Q')\sync{{\cal H}}\bar{H} $, i.e.,
$(P'\sync{L}Q')\setminus{\cal H} \approx_l (P''\sync{L}Q')\setminus{\cal H}$.
In the case that $Q'\transits{(h,r)} Q''$ the proof is analogous.
\end{itemize}

\end{proof}

Notice that the fact that $PSNI$ is not preserved by the choice operatior is a consequence of the fact that lumpable bisimilarity is not a congruence for this operator.

 We now prove that if $P\in PSNI$  then the equivalence class $[P]$ with respect to  lumpable bisimilarity $\approx_l$ is closed under \emph{PSNI}.

\begin{proposition}
Let $P$ and $Q$ be two PEPA components. If $P\in PSNI$ and $P\approx_l Q$ then also $Q\in PSNI$.
\end{proposition}
\begin{proof}
  Let $P\in PSNI$ such that $P\approx_l Q$. Let
  $ Q'\in ds(Q)$ such that $Q' \transits{(h,r)} Q''$. From the hypothesis that
  $P\approx_l Q$, there exist $P',P''\in ds(P)$ such that $P'\approx_l Q'$ and $P''\approx_l Q''$. Hence there exists $r'$ such that
  $P' \transits{(h,r')} P''$ and $P' \setminus {\cal H}\approx_l P''\setminus{\cal H}$. From the fact that $\approx_l$ is a congruence with respect to the cooperation operator we have
$Q' \setminus {\cal H}\approx_l Q''\setminus{\cal H}$ and then also $Q\in PSNI$.
  \end{proof}

\section{Comparison with other SOS-based persistent security properties}\label{sec:comparison}
The security property  presented in this paper is \emph{persistent} in the sense that if a model $P$ is secure then
all the states reachable by $P$ during its execution are also secure.
Persistence is not a common feature of Non-Interference properties.
For example, many properties based on trace models, like 
\emph{generalized Non-Inference} and \emph{separability} \cite{McL94},
and the \emph{non local
bisimulation based noninterference} properties for the Markovian process calculus
defined in \cite{AB09}
are  not persistent.
Persistence is used in 
  program verification techniques based on type-systems to provide
 sufficient conditions to
Non-Interference properties,
like, e.g., in
\cite{ABF18,HR03,SM03,SV98}.
In this setting persistence  provides  sufficient static
conditions which are invariant with respect to execution and imply the
desired dynamic property.

In \cite{jcs06}, a persistent property named \emph{P\_BNDC} has been proposed for 
non-deterministic CCS processes. The aim of this definition is to capture a robust notion 
of security for processes which may
 move in the middle of a computation.
In this context persistence ensures that a secure process always migrates to a secure state. Notice that
if the system satisfies a non-persistent property then it might migrate when it is executing in an insecure state
and then, from the point of view of the new host, the incoming process is insecure and, consequently,
it should not be executed.
As  our Persistent Stochastic Non-Interference property \emph{PSNI}, property \emph{P\_BNDC} is provided with two
  sound and complete characterizations: one in terms of a behavioural equivalence between processes up to high level contexts and
another one in terms of unwinding conditions. Let us compare the expressivity of
\emph{P\_BNDC} and \emph{PSNI} by considering their SOS-based characterization in terms of unwinding conditions.
The formal unwinding characterization of \emph{P\_BNDC} for CCS processes is the following:
\begin{definition}\label{def:pbndc}
Let $P$ be a CCS process and $H$ denote the set of all high level actions.
\[P\in P\_BNDC \mbox{ iff }
\, \forall P'\mbox{ reachable from P }, \, \] \[P' \transits{h} P'' \mbox{ implies }
P' \Transits{\hat{\tau}} P''' \mbox{ and }
P'' \setminus {H}\approx P'''\setminus{H}\]
where $\Transits{\hat{\tau}}$ represents a possibly empty sequence of $\tau$ transitions and $\approx$ denotes  Milner's weak bisimulation relation \cite{Mil89}.
\end{definition}
Both PEPA and CCS are provided with a  structural operational semantics
that  allows us to compare the  definitions of \emph{PSNI} (for PEPA processes) and 
\emph{P\_BNDC} (for CCS processes)  just by considering the processes label transition systems eventually removing information concerning the activity rates.
Consider for instance the simple process depicted in Figure \ref{fig:non-PSNI}. 
If we discard the activity rates we can interpret the graph as the labeled transition system of a CCS process $P$. According to Definition 
\ref{def:pbndc} we have that $P$ satisfies  \emph{P\_BNDC}. On the contrary, when we consider the activity rates we have the model of a PEPA process $P$ which, according to Theorem \ref{theo:unwinding}, does not satisfies \emph{PSNI}. Indeed we cannot find a lumpable bisimulation such that $P \setminus {\cal H}\approx_l P'\setminus{\cal H}$.

The unwinding definition of \emph{PSNI} resembles the definition of \emph{Strong BNDC (SBNC)} which has been introduced in \cite{cl-vmcai03} as a sufficient condition for verifying \emph{P\_BNDC}.

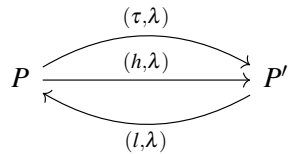
\begin{figure}[b]
\begin{center}
\begin{tikzcd}
   {} P\arrow[bend left]{rrr}{(\tau,\lambda)} \arrow{rrr}{(h,\lambda)} & & & P' \arrow[bend left]{lll}{(l,\lambda)} 
\end{tikzcd}
\end{center}
\caption{A simple two state model.}\label{fig:non-PSNI}
 \end{figure}

Recently, in \cite{AB09} Non-Interference properties for processes expressed as terms of a Markovian process calculus are introduced. The calculus presented in the paper allows the authors to model three kinds of actions: exponentially timed actions, immediate actions and passive actions. As a consequence, the proposed  process algebra encompasses nondeterminism, probability, priority and stochastic time. The behavioral observation defined by the authors extends the classical bisimulation relation of Milner \cite{Mil89}. The property named Bisimulation-based Strong Stochastic Local Non-Interference (\emph{BSSLNI}) is defined in the style of our unwinding conditions but it is based on an  observation equivalence named $\approx_{EMB}$ which abstracts from internal $\tau$ actions with zero duration. In particular, the relation $\approx_{EMB}$ is based on the idea that if a given class of processes is not reachable directly after executing a certain action, then one has to explore the possibility of reaching that class indirectly via a finite-length path $\pi$ of internal actions with zero duration but with a specific probability of execution $prob(\pi)$. As observed by the authors, in general the  performance indices of a system satisfying \emph{BSSLNI} are not independent from the presence or the absence of high level interactions.

\begin{figure}[tbp!]
\begin{center}
\begin{tikzcd}
{}& & &  P_1 \arrow[bend right]{llddd}[swap]{(h,\lambda)}  \arrow[bend left]{rrddd}{(l,\lambda)}  & & \\
\\
\\
   {}& P_2\arrow[bend right]{rrrr}[swap]{(l,\lambda)}&  & & & P_3 \arrow[bend left]{lluuu}{(l,\rho)}
\end{tikzcd}
\end{center}
\caption{A simple three state model.}\label{fig:PSNI}
\end{figure}
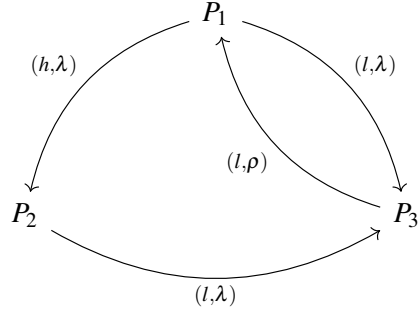

On the contrary, the observation equivalence at the base of our definition relies on the notion of lumpability and ensures that, for a secure process $P$,  the steady state probability of observing the system being in a specific state $P'$ is independent from its possible high level interactions. In order to show it consider  the simple three state system depicted in Figure \ref{fig:PSNI}.
In this case, following Theorem  \ref{theo:unwinding}, we can prove that $P_1\in PSNI$. Indeed, it is easy to prove that $P_1 \setminus {\cal H}\approx_l P_2\setminus{\cal H}$ when $\approx_l$ is the  lumpable bisimilarity.
In particular, the probability for a low level user to observe, in steady state, the system being in state
$P_3$ is independent from whether or not $P_1$ has performed the high level  activity $(h,\lambda)$.
To prove this, suppose that $P_1$  synchronizes over $h$. Then, for a low level observer, the system behaves as $P_1/{\cal H}$   depicted in Figure \ref{fig:hiding} $(a)$. We can compute the steady state distribution of $P_1/{\cal H}$ by solving the global balance equations together with the normalization condition, obtaining:

\[\begin{array}{rcl}
\pi_1*2\lambda& =& \pi_3*\rho\\
\pi_2*\lambda& =& \pi_1*\lambda\\
\pi_3*\rho& =& \pi_1*\lambda+\pi_2*\lambda\\
\pi_1+\pi_2+\pi_3& =& 1\\
\end{array}
\]
 whose solution is 
\[\begin{array}{cccccc}
\pi_1=\frac{\rho}{2(\lambda+\rho)}  & & \pi_2=\frac{\rho}{2(\lambda+\rho)} & &  \pi_3=\frac{\lambda}{\lambda+\rho} \\
\end{array}\]
where $\pi_1,\pi_2$ and $\pi_3$ denote the steady state probabilities of states $P_1/{\cal H}$, $P_2/{\cal H}$ and $P_3/{\cal H}$, respectively.

Consider now the case in which $P_1$ does not synchronize over $h$. Then the low level view of the system is represented by $P_1\setminus{\cal H}$ depicted in Figure \ref{fig:hiding} $(b)$. Again we can compute the steady state distribution of $P_1\setminus{\cal H}$ by solving the global balance equations together with the normalization condition, obtaining:

\[\begin{array}{rcl}
\pi_1*\lambda& =& \pi_3*\rho\\
\pi_3*\rho& =& \pi_1*\lambda\\
\pi_1+\pi_3& =& 1\\
\end{array}
\]
 whose solution is 
\[\begin{array}{cccccc}
\pi_1=\frac{\rho}{\lambda+\rho} & &  \pi_3=\frac{\lambda}{\lambda+\rho} \\
\end{array}\]
where $\pi_1$ and $\pi_3$ are the steady state probabilities of states $P_1\setminus{\cal H}$ and $P_3\setminus{\cal H}$ , respectively.
This proves that, from the low level point of view, the steady state probability of $P_3$ is independent from the fact that $P$ has cooperated with a high level context or not.

\begin{figure}[tbp!]
  \begin{center}
    \begin{small}
\begin{tikzcd}
{}& & &  P_1/{\cal H} \arrow[bend right]{llddd}[swap]{(\tau,\lambda)}  \arrow[bend left]{rrddd}{(l,\lambda)}    {}& & &   P_1\setminus{\cal H}  \arrow[bend left]{rrddd}{(l,\lambda)}  & & \\
\\
\\
  {}& P_2/{\cal H}\arrow[bend right]{rrrr}[swap]{(l,\lambda)}&  & & & P_3/{\cal H} \arrow[bend left]{lluuu}{(l,\rho)}  {}   & & & P_3\setminus{\cal H} \arrow[bend left]{lluuu}{(l,\rho)}\\
  \\
  \\
  {}&  & &  (a) \ \ P_1/{\cal H}    {} &  & & & (b) \ \ P_1\setminus{\cal H} & & \\
\end{tikzcd}
\caption{The models of $P_1/{\cal H}$ and  $P_1\setminus{\cal H}$.}\label{fig:hiding}
\end{small}
    \end{center}
\end{figure}
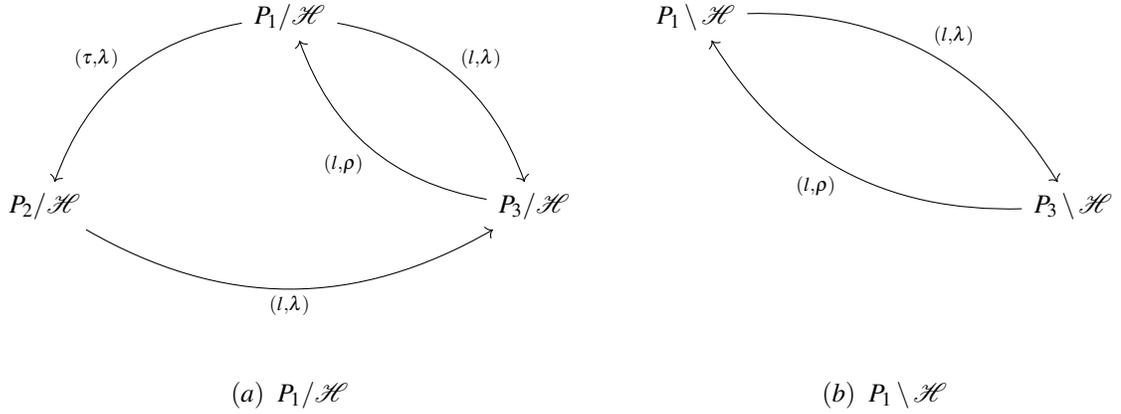

\section{Conclusion}\label{sec:conclusion}
In this paper we presented a \emph{persistent} information flow security property for stochastic processes expressed as terms of the PEPA process algebra.
Our property, named \emph{Persistent Stochastic Non-Interference} (\emph{PSNI}) is based on a structural operational semantics and a bisimulation based observation equivalence for the PEPA terms. We provide two characterizations for \emph{PSNI}: one in terms of a bisimulation-like equivalence relation and another one in terms of unwinding conditions.

The first characterization allows us to perform the verification of {\emph{PSNI} for finite state processes in polynomial time with respect to the number of states of the system \cite{cops}.

The second characterization is based on unwinding conditions. 
This kind of conditions for possibilistic security properties
have been already explored in the literature, like, e.g., in
\cite{RS01,Mil94,Man00b}.
Such unwinding conditions
have been  proposed for traces-based models
and represent only sufficient conditions
for their respective security properties. 
Differently, our unwinding conditions  provide both necessary and sufficient
 conditions for \emph{PSNI}.

Finally, in this paper we also deal with compositionality issues. Indeed, the
development of large and complex systems strongly depends on the ability
of dividing the task of the system into subtasks that are solved by system
subcomponents. Thus, it is useful to define properties which are compositional in the sense that
if the properties are satisfied by the system subcomponents then the system as a whole will satisfy the desired property
by construction.
We show that \emph{PSNI} is compositional with respect to low prefix, cooperation over low actions and hiding.

\vspace{1cm}
\textbf{Acknowledgments:} The work described in this paper has been partially supported by the Universit\`a Ca' Foscari Venezia - DAIS within the IRIDE program.
It has been also partially supported by the PRID project ENCASE financed by Universit\`a degli Studi di Udine and by GNCS-INdAM project Metodi Formali per la Verifica e la Sintesi di Sistemi Discreti e Ibridi.

\bibliographystyle{eptcs}
\bibliography{performance,security}

\end{document}